\documentclass[12pt,reqno,fleqn]{amsart}
\usepackage{latexsym}
\usepackage{amssymb}
\usepackage{amsxtra}
\usepackage{amsmath}
\usepackage{amsfonts}
\newtheorem{theorem}{Theorem}[section]
\newtheorem{lemma}[theorem]{Lemma}

\usepackage{amssymb}
\usepackage{mathrsfs}
\usepackage{amsmath}
\usepackage{enumerate}

\makeatother       % '@' is restored as a "non-letter" character for TeX

\makeatletter      % '@' is now a normail "letter" for TeX
\@addtoreset{equation}{section}
\makeatother       % '@' is restored as a "non-letter" character for TeX

\begin{document}

\title[gauge transformation of constrained discrete KP]{The gauge transformation of the constrained semi-discrete KP hierarchy}
\author{Maohua Li$^{1,2}$, Jipeng Cheng$^3$,  Jingsong
He$^2$$^*$} \dedicatory {
1. School of Mathematical Sciences, USTC, Hefei, 230026 Anhui, China;\\
2. Department of Mathematics, Ningbo University, Ningbo, 315211 Zhejiang, China;\\
3. Department of Mathematics,  CUMT, Xuzhou, 221116 Jiangsu,
 China }
\date{}

\thanks{$^*$ Corresponding author: hejingsong@nbu.edu.cn}

%%%%%%%%%%%%%%%%%%%%%%%%%%%%%%%%%%%%%%%%%%%%%%%%
\begin{abstract}
In this paper, the  gauge transformation of the constrained semi-discrete KP(cdKP) hierarchy
is constructed explicitly by the suitable choice of the generating functions.
Under the $m$-step successive gauge transformation $T_m$, we give the transformed (adjoint) eigenfunctions and the $\tau$-function of the transformed Lax operator of the cdKP hierarchy.
\end{abstract}

%%%%%%%%%%%%%%%%%%%%%%%%%%%%%%%%%%%%%%%%%%%%%%%%

\maketitle
%\tableofcontents \allowdisplaybreaks \setcounter{section}{0}

\section{Introduction}
The semi-discrete Kadomtsev-Petviashvili (dKP) hierarchy
\cite{Kupershimidt,Iliev,LiuS2,czli2012,zhangdj1,zhangdj2}  is  an attractive
research object in the field of the discrete integrable systems.
The dKP hierarchy is defined by means of the difference derivative
$\triangle$ instead of the usual derivative $\partial$ with respect
of $x$ in a classical system \cite{dkjm,dl1}, and the continuous
spatial variable is replaced by a discrete variable $n$.
By using a non-uniform shift of space variable, the $\tau$-function of KP hierarchy implies a special kind of $\tau$-function  for the
semi-discrete KP hierarchy \cite{Iliev}.
The ghost symmetry of the dKP hierarchy is constructed by using the additional
symmetry \cite{czli2012}. The dKP hierarchy possesses an infinite dimensional algebra  structure \cite{zhangdj1}.
 Very recently, the continuum limit of the dKP hierarchy is given
 in ref. \cite{zhangdj2}.

Gauge transformation is one kind of powerful method to construct the
solutions of the integrable systems for both the continuous KP
hierarchy \cite{chau_cmp1992,oevel1993,nimmo,oevelRMP93,MWX_LMP1997,hlc2002}
and the dKP hierarchy\cite{oe,LiuS}. The multi-fold of this transformation is expressed
directly by determinants \cite{hlc2002, LiuS}, and is used to construct multiple wave solutions to the generalized KP and
BKP equations\cite{MWX-AMC(2012)}. This transformation is also applicable to
the so-called constrained KP hierarchy \cite{wo1,anp3,cst1,wlg1,hlc2}.
It is known that there are two types of gauge transformation: differential type $T_d$
and integral type $T_i$. Because of the reduction conditions of the BKP hierarchy and the CKP
hierarchy, it is necessary to consider the pair of $T_d$ and $T_i$. This has been
used to construct the two-peak soliton \cite{hcr1} and to construct the gauge transformation of the
constrained BKP hierarchy and the constrained CKP hierarchy \cite{hwc1}.  And it is known that the KP hierarchy has been generalized to constrained flows and extended flows with self-consistent sources \cite{MWX2010}.
It is interesting to note that
the determinant representation of the dressing transformation of the
extended two-dimensional Toda lattice hierarchy is also developed \cite{liuxijun}.

Similar to the constrained KP hierarchy \cite{{kss,chengyiliyishenpla1991,
chengyiliyishenjmp1995}}, the constrained semi-discrete KP(cdKP) hierarchy
\cite{lmh2} is defined by $\frac{\partial L}{\partial t_l} = [(L^l)_+,L]$ with a Lax operator
$L =\triangle + \sum_{i=1}^{m}q_i(t)\triangle^{-1}r_i(t)$.
 The so called discrete non-linear Schr\"{o}dinger (generalized DNLS) equation \cite{ablowitz04}
\begin{eqnarray*}
 q_{1,t_2}=\triangle^2 q_1+2q_1^2r_1,\\
r_{1,t_2}=-{\triangle^*}^2 r_1+2q_1r_1^2,
\end{eqnarray*}
can be generated from the $t_2$ flows of cdKP hierarchy. Further the
additional symmetry of the cdKP hierarchy is constructed in
ref. \cite{lmh2}, which shows that the constraint in the dKP hierarchy
preserves the symmetry structures with very minor modification.
However the gauge transformation of the cdKP hierarchy  has not
appeared in literatures.

 The purpose of this paper is to construct the gauge transformation of the cdKP hierarchy.
 As we shall show, it is not a trivial task to reduce the gauge transformation of the dKP hierarchy
 to the cdKP hierarchy. For the cdKP hierarchy, the transformed eigenfunctions and the adjoint
 eigenfunctions can not be conserved for the originally form. If the generating functions of
 gauge transformation of the cdKP hierarchy are  selected from the
(adjoint)eigenfunctions, we get the $\Delta$-Wronskian representation of the transformed $\tau$
function.

This paper is organized as follows. Some basic results  of the dKP
hierarchy and the cdKP hierarchy are summarized in Section
\ref{section2}. After introducing of two types gauge transformations
of the cdKP hierarchy, the transformation rules of the eigenfunction
functions and the $\tau$ function of the cdKP hierarchy are obtained
by means of a crucial modification from the transformation of the
dKP hierarchy in Section \ref{section3}. Next the successive
applications of the difference type gauge transformation have be
discussed in Section \ref{section4}. And we establish the
determinant representation of gauge transformation operator $T_m$,
then obtained a general form of the $\tau$-function
$\tau^{(m)}_{\triangle}$ for the cdKP hierarchy. Section
\ref{conclusion} is devoted to conclusions and discussions.

\section{the constrained semi-discrete KP hierarchy}\label{section2}

Let us briefly recall some basic facts
about the semi-discrete KP (cdKP) hierarchy according to reference \cite{Iliev}.
Firstly
a space $F$, namely
\begin{equation}
F=\left\{ f(n)=f(n,t_1,t_2,\cdots,t_j,\cdots);  n\in\mathbb{Z}, t_i\in\mathbb{R}
\right\}
\end{equation}
is defined for the space of the semi-discrete KP hierarchy.
$\Lambda$ and $\triangle$ are denote for the shift operator and the
difference operator, respectively. Their actions on function $f(n)$ are
defined as
\begin{equation}\Lambda f(n)=f(n+1)
\end{equation}
and
\begin{equation}\triangle f(n)=f(n+1)-f(n)=(\Lambda -I)f(n)
\end{equation} respectively,
where $I$ is the identity operator.

For any $j\in\mathbb{Z},$ the Leibniz  rule of $\triangle$ operation is,

\begin{equation}\triangle^j\circ
f=\sum^{\infty}_{i=0}\binom{j}{i}(\triangle^i
f)(n+j-i)\triangle^{j-i},\hspace{.3cm}
\binom{j}{i}=\frac{j(j-1)\cdots(j-i+1)}{i!}.\label{81}
\end{equation}
So an associative ring $F(\triangle)$ of formal pseudo
difference operators is obtained, with the operation $``+"$ and $``\circ"$, namely
$F(\triangle)=\left\{R=\sum_{j=-\infty}^d f_j(n)\triangle^j,
f_j(n)\in R, n\in\mathbb{Z}\right\}
$. The
adjoint operator to the $\triangle$ operator is given by
$\triangle^*$,
\begin{equation}
\triangle^* \circ f(n)=(\Lambda^{-1}-I)f(n)=f(n-1)-f(n),
\end{equation}
where $\Lambda^{-1} f(n)=f(n-1)$, and the corresponding ``$\circ$"
operation is
\begin{equation}
\triangle^{*j}\circ
f=\sum^{\infty}_{i=0}\binom{j}{i}(\triangle^{*i}f)(n+i-j)\triangle^{*j-i}.
\end{equation}
Then the adjoint ring $F(\triangle^*)$ to the
$F(\triangle)$ is obtained, and the formal adjoint to $R\in F(\triangle)$ is defined
by $R^*\in F(\triangle^*)$ as $R^*=\sum_{j=-\infty}^d
\triangle^{*j}\circ f_j(n)$. The $"*"$ operation satisfies the rules as
$(F\circ G)^*=G^*\circ F^*$ for two operators $F$ and $G$ and $f(n)^*=f(n)$ for
a function $f(n)$.

We list some useful properties for the difference operators as following:
\begin{lemma}
For $f\in F$, $\triangle$ and $\Lambda$ as
above, the following identities hold.
\begin{align}
&(1)\quad \triangle\circ\Lambda=\Lambda\circ\triangle,  \\
&(2)\quad \triangle^*=-\triangle\circ\Lambda^{-1},  \\
&(3)\quad
(\triangle^{-1})^*=(\triangle^*)^{-1}=-\Lambda\circ\triangle^{-1}, \\
&(4)\quad \triangle^{-1}\circ
f\circ\triangle^{-1}=(\triangle^{-1}f)\circ\triangle^{-1}-\triangle^{-1}\circ
\Lambda(\triangle^{-1} f),\label{85}\\
&(5)\quad\triangle\circ f(n)=\Lambda(f(n))\circ\triangle + \triangle(f(n)).\label{gs2}
\end{align}
\end{lemma}

 The so-called $1$-constrained semi-discrete KP (cdKP) hierarchy \cite{lmh2} is defined by the following Lax equation
\begin{equation}\label{1cdKPlaxeq}
\frac{\partial L}{\partial t_l} = [(L^l)_+,L], l=1,2,\cdots,
\end{equation}
associated with a special Lax operator
\begin{equation} \label{laxofcdkp}
L = L_{+}+ \sum_{i=1}^{m}q_i(t)\triangle^{-1}r_i(t)=\triangle + \sum_{i=1}^{m}q_i(t)\triangle^{-1}r_i(t),
\end{equation}
 and  $q_i(t)$ is an eigenfunction, $r_i(t)$ is an adjoint
 eigenfunction of the Lax operator $L$. The eigenfunction and adjoint eigenfunction {$q_i(t),r_i(t)$} are
important dynamical variables in the cdKP hierarchy.
One can check that the
 Lax equation (\ref{1cdKPlaxeq}) is consistent with the evolution
 equations of the eigenfunction(or adjoint eigenfunction)
\begin{eqnarray}\label{eigenfunction}
\begin{cases}
q_{i,t_m} = B_mq_i,\\
r_{i,t_m} = -B_m^*r_i,\quad  B_m=(L^m)_{+}, \forall m \in N.
\end{cases}
\end{eqnarray}
Therefore  the cdKP hierarchy in eq.(\ref{1cdKPlaxeq}) is well defined.
 From the Lax equation (\ref{1cdKPlaxeq}), we get the first nontrival $t_2$ flow equations of the cdKP hierarchy for $m=1,l=2$ as
\begin{eqnarray}
\begin{cases}
q_{1,t_2}=\triangle^2 q_1+2q_1^2r_1=q_1(n+2)-2q_1(n+1)+q_1(n)+2q_1^2r_1,\\
r_{1,t_2}=-{\triangle^*}^2 r_1+2q_1r_1^2=r_1(n)-2r_1(n-1)+r_1(n-2)+2q_1(n)r_1(n)^2
\end{cases}
\end{eqnarray}
And it is so called the generalized discrete non-linear Schr\"{o}dinger (generalized DNLS) equation \cite{ablowitz04}. It can be reduced to the discrete non-linear Schr\"{o}dinger (DNLS) equation \cite{ablowitz04} by letting $r_1=q_1^*$ and a scaling transformation $t_2=it_2$.

\section{Gauge transformations of  the constrained semi-discrete KP hierarchy}\label{section3}
We will discuss the gauge transformations of the constrained semi-discrete KP hierarchy in
this section. It is reported two types of gauge transformation operators  for the semi-discrete KP hierarchy in \cite{LiuS}. We will extended the gauge transformation to the constrained semi-discrete  KP hierarchy. If there exist a pseudo-difference operator $T$  satisfying
\begin{equation}\label{TLT-1}
L^{(1)}=T\circ L \circ T^{-1}, B_n^{(1)}=(L^{(1)})_+^n,
\end{equation}
so that
\begin{equation}
\frac{\partial L}{\partial t_l} = [(L^l)_+,L]\nonumber
\end{equation}
holds for the transformed Lax operator $L^{(1)}$, i.e.,
\begin{equation}\label{L1}
\frac{\partial L^{(1)}}{\partial t_l} = [(L^{(1)})^{l}_+,L^{(1)}];
\end{equation} then $T$ is called a gauge transformation operator of the cdKP hierarchy.
According to the definition of gauge transformation, we have the following criterion lemma.
\begin{lemma}\label{td}
The operator $T$ is a gauge transformation operator, if
\begin{equation}\label{td1}
(T\circ B_n \circ T^{-1})_+=T\circ B_n \circ T^{-1} +\frac{\partial T}{\partial t_n}\circ T^{-1},
\end{equation}
or
\begin{equation}\label{td2}
(T\circ B_n \circ T^{-1})_-=-\frac{\partial T}{\partial t_n}\circ T^{-1}.
\end{equation}
\end{lemma}

Similar to the KP hierarchy and the cKP hierarchy, there  are two types of gauge transformation
operators of the cdKP hierarchy  as the following lemma:
\begin{lemma}
\cite{chau_cmp1992,oe}The cdKP hierarchy have two types  gauge transformation operators, namely,
\begin{eqnarray}
(1).T_d(q)&=&\Lambda(q)\circ\triangle\circ q^{-1},\label{gfofcdkp1}\\
(2).T_i(r)&=&\Lambda^{-1}(r^{-1})\circ\triangle^{-1}\circ r.\label{gfofcdkp2}
\end{eqnarray}
Where $q$ and $r$ are defined by (\ref{eigenfunction}) that are the (adjoint) eigenfunction of $L$ in (\ref{laxofcdkp}), which is called the generating functions of gauge transformation.
\end{lemma}

Via the gauge transformations of two types, $L^{(0)}=L$ becomes $L^{(1)}$ by the following lemma.
\begin{theorem}\label{3.3}
Under the gauge transformation of $T_d(q)$, the transformed Lax operator reads as
\begin{eqnarray}
L^{(1)}&=& L^{(1)}_++L^{(1)}_-,\\
L^{(1)}_+&=&\Lambda(L^{(0)}_+)+\Lambda(q)\circ\triangle(q^{-1}L^{(0)}_+q)_{\geq1}\circ \triangle^{-1}\circ \Lambda(q^{-1}),\\
L^{(1)}_-&=&q_0^{(1)}\triangle^{-1}r_0^{(1)}+\sum_{i=1}^{m}q_i^{(1)}\triangle^{-1}r_i^{(1)},\\
q_0^{(1)}&=&T_d(q)(L^{(0)})(q),r_0^{(1)}=\Lambda(q^{-1}),\\
q_i^{(1)}&=&T_d(q)q_i^{(0)},r_i^{(1)}=(T_d^{-1})^{*}(q)(r_i^{(0)}).
\end{eqnarray}
\end{theorem}
\begin{proof}
\begin{eqnarray*}
L^{(1)}_+&=&(T_d(q)\circ L^{(0)}\circ T^{-1}_d(q))_+\\
&=&(\Lambda(q)\circ\triangle\circ q^{-1}\circ L^{(0)}_{+}\circ q\circ \triangle^{-1} \circ\Lambda(q^{-1}))_+\\
&=&(\Lambda(q)\circ\Lambda( q^{-1}L^{(0)}_{+}q)\circ\Lambda(q^{-1}))_++(\Lambda(q)\triangle(q^{-1} L^{(0)}_{+}q)\circ \triangle^{-1} \circ\Lambda(q^{-1}))_+\\
&=&\Lambda(L^{(0)}_{+})+\Lambda(q)\triangle(q^{-1} L^{(0)}_{+}q)_{\geq1}\triangle^{-1} \circ\Lambda(q^{-1}),
\end{eqnarray*}
where used the identity (\ref{gs2}).
\begin{eqnarray*}
L^{(1)}_-&=&(T_d(q)\circ L^{(0)}\circ T^{-1}_d(q))_-\\
&=&(\Lambda(q)\triangle\circ q^{-1}L^{(0)}_{+}\circ q\circ \triangle^{-1} \circ\Lambda(q^{-1}))_-\\
&&+(\Lambda(q)\triangle\circ q^{-1}\circ\sum_{i=1}^{m}q_i^{(0)}\triangle^{-1}r_i^{(0)}\circ q\circ \triangle^{-1} \circ\Lambda(q^{-1}))_-.
\end{eqnarray*}
Where
\begin{eqnarray*}
&&(\Lambda(q)\triangle\circ q^{-1}L^{(0)}_{+}\circ q\circ \triangle^{-1} \circ\Lambda(q^{-1}))_-\\
&&=(T_d(q)L^{(0)}_{+}\circ q)_0\circ \triangle^{-1} \circ\Lambda(q^{-1})\\
&&=T_d(q)L^{(0)}_{+}(q)\circ \triangle^{-1} \circ\Lambda(q^{-1}),
\end{eqnarray*}
and
\begin{eqnarray*}
&&(\Lambda(q)\triangle\circ q^{-1}\circ\sum_{i=1}^{m}q_i^{(0)}\triangle^{-1}r_i^{(0)}\circ q\circ \triangle^{-1}\circ\Lambda(q^{-1}))_-\\
&&\stackrel{(\ref{85})}{==}\sum_{i=1}^{m}T_d(q)\circ q_i^{(0)}\cdot\triangle^{-1}(r_i^{(0)}q) \triangle^{-1}\circ\Lambda(q^{-1})\\
&&-\sum_{i=1}^{m}T_d(q)\circ q_i^{(0)}\triangle^{-1}\circ\Lambda(\triangle^{-1}(r_i^{(0)}q))\cdot\Lambda(q^{-1}))\\
&&=\sum_{i=1}^{m}T_d(q)(q_i^{(0)}\triangle^{-1}r_i^{(0)})(q)\triangle^{-1}\circ\Lambda(q^{-1})+\sum_{i=1}^{m}T_d(q)( q_i^{(0)})\triangle^{-1}\circ(\triangle^{-1})^*(r_i^{(0)}q)\cdot\Lambda(q^{-1})).
\end{eqnarray*}
When the above  two formulas are substituted into $L^{(1)}_-$, then
\begin{eqnarray*}
L^{(1)}_-&=&T_d(q)L^{(0)}_{+}(q)\circ \triangle^{-1} \circ\Lambda(q^{-1})+\sum_{i=1}^{m}T_d(q)(q_i^{(0)}\triangle^{-1}r_i^{(0)})(q)\triangle^{-1}\circ\Lambda(q^{-1})\\
&&+\sum_{i=1}^{m}T_d(q)( q_i^{(0)})\triangle^{-1}\circ(\triangle^{-1})^*(r_i^{(0)}q)\cdot\Lambda(q^{-1}))\\
&=&T_d(q)(L^{(0)}_{+}+\sum_{i=1}^{m}q_i^{(0)}\triangle^{-1}r_i^{(0)})(q)\circ \triangle^{-1} \circ\Lambda(q^{-1})\\
&&+\sum_{i=1}^{m}T_d(q)( q_i^{(0)})\triangle^{-1}\circ\Lambda(q^{-1})(\triangle^{-1})^*(q r_i^{(0)})\\
&=&T_d(q)L^{(0)}(q)\circ \triangle^{-1} \circ\Lambda(q^{-1})+\sum_{i=1}^{m}T_d(q)( q_i^{(0)})\triangle^{-1}\circ (T_d^{-1}(q))^*(r_i^{(0)}).
\end{eqnarray*}

If let $q_0^{(1)}=T_d(q)(L^{(0)})(q)$, $r_0^{(1)}=\Lambda(q^{-1})$, $q_i^{(1)}=T_d(q)q_i^{(0)}$, $r_i^{(1)}=(T_d^{-1})^{*}(q)(r_i^{(0)})$, we can get this theorem.
\end{proof}

\begin{theorem}
Under the type $2$ gauge transformation $T_i(r)$, the transformed Lax operator reads
\begin{eqnarray}
L^{(1)}&=& L^{(1)}_++L^{(1)}_-,\\
L^{(1)}_+&=&\Lambda^{-1}(L^{(0)}_+)-\Lambda^{-1}(r^{-1})\triangle^{-1}\circ\triangle^{*}(r\circ L^{(0)}_+\circ r^{-1})_{\geq1}\circ \Lambda(r^{-1}),\\
L^{(1)}_-&=&q_0^{(1)}\triangle^{-1}r_0^{(1)}+\sum_{i=1}^{m}q_i^{(1)}\triangle^{-1}r_i^{(1)},\\
q_0^{(1)}&=&\Lambda^{-1}(r^{-1}),
r_0^{(1)}=(T_i^{-1}(r))^{*}(L^{(0)})^*(r),\\
q_i^{(1)}&=&T_i(r)(q_i^{(0)}),
r_i^{(1)}=(T_i^{-1}(r))^{*}(r_i^{(0)}).
\end{eqnarray}
\end{theorem}
\begin{proof}
\begin{eqnarray*}
L^{(1)}_+&=&(T_i(r)\circ L^{(0)}\circ T^{-1}_i(r))_+\\
&=&(\Lambda^{-1}(r^{-1})\circ\triangle^{-1}\circ r\circ L^{(0)}_{+}\circ r^{-1}\circ \triangle\circ\Lambda^{-1}(r))_+\\
&=&(\Lambda^{-1}(r^{-1})\circ\Lambda^{-1}(rL^{(0)}_{+}r^{-1})\circ\Lambda^{-1}(r))_+-(\Lambda^{-1}(r^{-1})\circ\triangle^{-1}\circ \triangle(\Lambda^{-1}(rL^{(0)}_{+}r^{-1}))\circ\Lambda^{-1}(r))_+\\
&=&\Lambda^{-1}(L^{(0)}_+)+(\Lambda^{-1}(r^{-1})\triangle^{-1}\circ\triangle^{*}(r^{-1}\circ (L^{(0)}_+)^{*}\circ r)^*\circ \Lambda(r^{-1}))_+\\
&=&\Lambda^{-1}(L^{(0)}_+)+\Lambda^{-1}(r^{-1})\triangle^{-1}\circ\triangle^{*}(r(L^{(0)}_+)^{*}\circ r^{-1})_{\geq1}\Lambda(r^{-1}),
\end{eqnarray*}
where used the identity (\ref{gs2}).
\begin{eqnarray*}
L^{(1)}_-&=&[T_i(r)\circ(L^{(0)}_{+} + \sum_{i=1}^{m}q_i^{(0)}\triangle^{-1}r_i^{(0)})\circ T_i^{-1}(r)]_-\\
&=&(\Lambda^{-1}(r^{-1})\triangle^{-1}\circ r L^{(0)}_{+}\circ r^{-1}\circ\triangle\circ\Lambda(r))_-\\
&&+(\Lambda^{-1}(r^{-1})\triangle^{-1}\circ r(\sum_{i=1}^{m}q_i^{(0)}\triangle^{-1}r_i^{(0)}) r^{-1} \triangle \circ\Lambda^{-1}(r))_-.
\end{eqnarray*}
Where
\begin{eqnarray*}
&&(\Lambda^{-1}(r^{-1})\triangle^{-1}\circ r L^{(0)}_{+}\circ r^{-1}\circ\triangle\circ\Lambda^{-1}(r))_-\\
&=&(\Lambda^{-1}(r^{-1})\circ\Lambda^{-1}(rL^{(0)}_{+}r^{-1})\circ\Lambda^{-1}(r))_--(\Lambda^{-1}(r^{-1})\circ\triangle^{-1}\circ \triangle(\Lambda^{-1}(rL^{(0)}_{+}r^{-1}))\circ\Lambda^{-1}(r))_-\\
&=&(\Lambda^{-1}(L^{(0)}_+))_--(\Lambda^{-1}(r^{-1})\circ\triangle^{-1}\circ \Lambda^{-1}(r)\triangle\Lambda^{-1}((rL^{(0)}_{+}r^{-1})^*)^*)_-\\
&=&\Lambda^{-1}(r^{-1})\circ\triangle^{-1}\circ (T_i^{-1}(r))^{*}(L^{(0)}_{+})^{*} (r),
\end{eqnarray*}
and
\begin{eqnarray*}
&&(\Lambda^{-1}(r^{-1})\triangle^{-1}\circ r(\sum_{i=1}^{m}q_i^{(0)}\triangle^{-1}r_i^{(0)}) r^{-1} \triangle \circ\Lambda^{-1}(r))_-\\
&\stackrel{(\ref{85})}{==}&\sum_{i=1}^{m}\Lambda^{-1}(r^{-1})\triangle^{-1}\circ\triangle(\triangle^{-1}(r q_i^{(0)})\cdot\Lambda^{-1}(r_i^{(0)}r^{-1})) \Lambda^{-1}(r)\\
&&-\sum_{i=1}^{m}\Lambda^{-1}(r^{-1})\triangle^{-1}(r q_i^{(0)})\circ\triangle^{-1}\circ\triangle(\Lambda^{-1}(r_i^{(0)}r)) \Lambda^{-1}(r)\\
&=&\sum_{i=1}^{m}T_i(r)(q_i^{(0)})\cdot\triangle^{-1}\circ (T_i^{-1})^*(r)(r_i^{(0)})\\
&&+\Lambda^{-1}(r^{-1})\triangle^{-1}\circ (T_i^{-1}(r))^{*} (\sum_{i=1}^{m}q_i^{(0)}\triangle^{-1}r_i^{(0)})^*(r).
\end{eqnarray*}
When the above  two formulas are substituted into $L^{(1)}_-$, then
\begin{eqnarray*}
L^{(1)}_-&=&\Lambda^{-1}(r^{-1})\circ\triangle^{-1}\circ (T_i^{-1}(r))^{*}(L^{(0)}_{+})^{*} (r)
+\Lambda^{-1}(r^{-1})\triangle^{-1}\circ (T_i^{-1}(r))^{*} (\sum_{i=1}^{m}q_i^{(0)}\triangle^{-1}r_i^{(0)})^*(r)\\
&&+\sum_{i=1}^{m}T_i(r)(q_i^{(0)})\cdot\triangle^{-1}\circ (T_i^{-1})^*(r)(r_i^{(0)})\\
&=&\Lambda^{-1}(r^{-1})\circ\triangle^{-1}\circ (T_i^{-1}(r))^{*}(L^{(0)})^{*} (r)+\sum_{i=1}^{m}T_i(r)(q_i^{(0)})\cdot\triangle^{-1}\circ (T_i^{-1})^*(r)(r_i^{(0)}).
\end{eqnarray*}
If let $q_0^{(1)}=\Lambda^{-1}(r^{-1})$, $r_0^{(1)}=(T_i^{-1}(r))^{*}(L^{(0)})^*(r)$, $q_i^{(1)}=T_i(r)(q_i^{(0)})$ and $r_i^{(1)}=(T_i^{-1}(r))^{*}(L^{(0)})^*(r)$, then
\begin{eqnarray*}
L^{(1)}_-&=&q_0^{(1)}\triangle^{-1}r_0^{(1)}+\sum_{i=1}^{m}q_i^{(1)}\triangle^{-1}r_i^{(1)}.
\end{eqnarray*}
\end{proof}

\textbf{Remark:}
Although the transformations (\ref{tau1q}) and (\ref{tau1r}) look like as the same as the formula in the semi-discrete KP hierarchy, but there are  main difference between the dKP hierarchy and the cdKP hierarchy. We can see the number of the (adjoint) eigenfunctions has added one after each time of gauge transformation. So for the cdKP hierarchy, to ensure  that the gauge transformed Lax operator preserves the (\ref{laxofcdkp}), it can be  $q_i^{(1)}\triangle^{-1}r_i^{(1)}=0$ for some one $i$. So the generating function $q,r$ of the gauge transformation operator $T_d(q)$ and $T_i(r)$ can not be arbitrarily chosen. This theorem means there are two choices to keep the form of the Lax operator of cdKP hierarchy. They must be selected from the eigenfunction $q_i$ and the adjoint eigenfunction $r_i$ respectively.
And the operator $T_d(q_i)$ and $T_i(r_i)$ will annihilate their generation functions, i.e. ,
\begin{eqnarray*}
T_d(q_i)(q_i)=0,
(T_i^{-1}(r_i))^{*}(r_i)=0,i=1,2,\dots,m.
\end{eqnarray*}
If  the generating function of the gauge transformation in theorem \ref{3.3} was selected for $q_1$, then $q_1^{(1)}=T_d(q_1)(q_1)=0$. And $q_0^{(1)}$ takes over its role.

\begin{theorem}
(a).Under the gauge transformation $L^{(1)}=T_d(q_1)\circ L^{(0)}\circ T_d^{-1}(q_1)$,  the eigenfunction $q_i^{(0)}=q_i$ and  adjoint eigenfunction $r_i^{(0)}=r_i$ of $L^{(0)}=L$ are transformed into new eigenfunction $q_i^{(1)}$ and new adjoint eigenfunction $r_i^{(1)}$ of $L^{(1)}$ by
\begin{align}
&q_1^{(1)}=T_d(q_1^{(0)})(L^{(0)})(q_1), r_1^{(1)}=\Lambda(q_1^{-1}),\\
&q_i^{(1)}=T_d(q_1)q_i^{(0)}, r_i^{(1)}=(T_d^{-1})^{*}(q_1)(r_i^{(0)}),i=2,\cdots,m,
\end{align}
and the $\tau$ function $\tau^{(0)}_{\triangle}$ of $L^{(0)}$ is transformed into the new $\tau$ function $\tau^{(1)}_{\triangle}$ of $L^{(1)}$ by
\begin{equation}\label{tau1q}
\tau^{(1)}_{\triangle}=q_1\tau^{(0)}_{\triangle}.
\end{equation}
(b).Under the gauge transformation $L^{(1)}=T_i(r_1)\circ L^{(0)}\circ T_i^{-1}(r_1)$,  the eigenfunction $q_i^{(0)}=q_i$ and  adjoint eigenfunction $r_i^{(0)}=r_i$ of $L^{(0)}=L$ are transformed into new eigenfunction $q_i^{(1)}$ and new adjoint eigenfunction $r_i^{(1)}$ of $L^{(1)}$ by
\begin{align}
&q_1^{(1)}=\Lambda^{-1}(r_1^{-1}),
r_1^{(1)}=(T_i^{-1}(r_1))^{*}(L^{(0)})^*(r_1),\\
&q_i^{(1)}=T_i(r)q_i^{(0)},
r_i^{(1)}=(T_i^{-1})^{*}(r)(r_i^{(0)}),i=2,\cdots,m,
\end{align}
and the $\tau$ function $\tau^{(0)}_{\triangle}$ of $L^{(0)}$ is transformed into the new $\tau$ function $\tau^{(1)}_{\triangle}$ of $L^{(1)}$ by
\begin{equation}\label{tau1r}
\tau^{(1)}_{\triangle}=\Lambda^{-1}(r_1)\tau^{(0)}_{\triangle}.
\end{equation}
\end{theorem}

\section{Successive applications of gauge transformations}\label{section4}

In order to investigate the new result of successive transformations by using the
  gauge transformation operators, we will discuss successive
 applications of the difference gauge transformation operator $T_d$, which is
 like to the classical case \cite{cst1,hlc2}. Firstly, we only consider the
  chain of gauge transformation operator of single-channel \cite{cst1} difference type $T_d(q_1)$ starting from the initial Lax operator $L^{(0)}=L$,
\begin{equation}\label{succesT}
L^{(0)}\xrightarrow{T_d^{(1)}(q_1^{(0)})} L^{(1)}\xrightarrow{T_d^{(2)}(q_1^{(1)})} L^{(2)}\xrightarrow{T_d^{(3)}(q_1^{(2)})} L^{(3)}\rightarrow \dots \rightarrow L^{(n-1)}\xrightarrow{T_d^{(n)}(q_1^{(n-1)})} L^{(n)}.
\end{equation}
Here the index $i$ in the gauge transformation operator $T_d^{(i)}(q_1^{(j-1)})$ means the $i$-th gauge transformation, and $q_1^{(j)}$ (or  $r_1^{(j)}$) is transformed by $j$-steps gauge transformations from $q_1$ (or  $r_1$), $L^{(k)}$ is transformed by $k$-step gauge transformations from the initial Lax operator $L$.

 Now we firstly consider successive gauge transformations in (\ref{succesT}). We define the operator as
\begin{equation}\label{TMK}
T_{m}=T_d^{(m)}(q_1^{(m-1)})\circ \dots \circ T_d^{(2)}(q_1^{(1)})\circ T_d^{(1)}(q_1^{(0)}),
\end{equation}
in which
\begin{eqnarray}
q_i^{(j)}=T_d^{(j)}(q_1^{(j-1)})\circ \dots \circ T_d^{(2)}(q_1^{(1)})\circ T_d^{(1)}(q_1^{(0)}) q_i,i,j=1,\cdots,m;\\
r_k^{(j)}=((T_d^{(j)})^{-1})^{*}(q_1^{(j-1)})\circ \dots \circ ((T_d^{(2)})^{-1})^{*}(q_1^{(1)})\circ ((T_d^{(1)})^{-1})^{*}(q_1^{(0)}) r_k,j,k=1,\cdots,m.
\end{eqnarray}

In order to express the determinant representation of $T_m$, we would like to define the generalized discrete $\triangle$-Wronskian for the eigenfunctions $\{q_1,q_2,\dots,q_m\}$ of $L$ as
\begin{equation}
W_m^{\triangle}(q_1,q_2,\dots,q_m) =
\left|
\begin{array}{cccc}
 q_1 &q_2 &  \cdots  &q_m \\
  \triangle q_1 & \triangle q_2 &  \cdots  & \triangle q_m \\
    \vdots  &  \vdots  &  \ddots  &  \vdots   \\
   \triangle^{m-1} q_1 &\triangle^{m-1}q_2 &  \cdots  & \triangle^{m-1}q_m
   \end{array}
   \right|,
\end{equation}
\begin{equation}
IW_{m+1}^{\triangle}(q_1,q_2,\dots,q_m)=\left|
\begin{array}{ccccc}
 q_1\circ \triangle^{-1} &\Lambda( q_1)&\Lambda(\triangle q_1) &  \cdots  &\Lambda(\triangle^{m-2} q_1)\\
  q_2\circ \triangle^{-1} &\Lambda(q_2) &\Lambda(\triangle q_2) &  \cdots  & \Lambda(\triangle^{m-2} q_2) \\
    \vdots  &  \vdots &  \vdots  &  \ddots  &  \vdots   \\
   q_m\circ \triangle^{-1} &\Lambda( q_m)&\Lambda(\triangle q_m) &  \cdots  & \Lambda(\triangle^{m-2} q_m)
   \end{array}
   \right|.
\end{equation}
Using the gauge transformation operator $T_d(q_1)$, the $m$-step gauge transformation can be construct for
\begin{equation}\label{}
L^{(0)}\xrightarrow{T_d^{(1)}(q_1^{(0)})} L^{(1)}\xrightarrow{T_d^{(2)}(q_1^{(1)})} L^{(2)}\xrightarrow{T_d^{(3)}(q_1^{(2)})} L^{(3)}\rightarrow \dots \rightarrow L^{(m-1)}\xrightarrow{T_d^{(m)}(q_1^{(m-1)})} L^{(m)}.
\end{equation}

If $\eta_i$ is defined by
\begin{equation}\label{eta}
\eta_{i+1}\triangleq(L^{(0)})^i\cdot q_1^{(0)},
\end{equation} $\eta_i^{(j)}$ is  the $j$-step transformed form from $\eta_i$.
It is easy got $\eta_{i+1}^{(i)}=q_1^{(i)},i=1,\cdots,m,$ by the mathematical induction.

\begin{theorem}\label{theor1}
\cite{LiuS}The gauge transformation operator $T_m$ and $T_m^{-1}$ have the following determinant representation:
\begin{eqnarray}\label{TM}
T_m&=&T_d^{(m)}(\eta_m^{(m-1)})\circ \dots \circ T_d^{(2)}(\eta_2^{(1)})\circ T_d^{(1)}(\eta_1^{(0)})\\\nonumber
&=&\frac{1}{W_m^{\triangle}(\eta_1,\eta_2,\dots,\eta_m)}
\left|
\begin{array}{ccccc}
 \eta_1 &\eta_2 &  \cdots  &\eta_m &1\\
  \triangle \eta_1 & \triangle \eta_2 &  \cdots  & \triangle \eta_m &\triangle\\
    \vdots  &  \vdots  &  \vdots &  \ddots  &  \vdots   \\
   \triangle^{m-1} \eta_1 &\triangle^{m-1}\eta_2 &  \cdots  & \triangle^{m-1}\eta_m & \triangle^{m-1}\\
   \triangle^{m} \eta_1 &\triangle^{m}\eta_2 &  \cdots  & \triangle^{m}\eta_m & \triangle^{m}
   \end{array}
   \right|,
\end{eqnarray}
and
\begin{equation}
T_m^{-1}=\left|
\begin{array}{ccccc}
 \eta_1\circ \triangle^{-1} &\Lambda( \eta_1)&\Lambda(\triangle \eta_1) &  \cdots  &\Lambda(\triangle^{m-2} \eta_1)\\
  \eta_2\circ \triangle^{-1} &\Lambda(\eta_2) &\Lambda(\triangle \eta_2) &  \cdots  & \Lambda(\triangle^{m-2} \eta_2) \\
    \vdots  &  \vdots &  \vdots  &  \ddots  &  \vdots   \\
   \eta_m\circ \triangle^{-1} &\Lambda( \eta_m)&\Lambda(\triangle \eta_m) &  \cdots  & \Lambda(\triangle^{m-2} \eta_m)
   \end{array}
   \right|
   \frac{(-1)^{m-1}}{\Lambda(W_m^{\triangle}(\eta_1,\eta_2,\dots,\eta_m))}.
\end{equation}
Here the determinant of $T_m$ is expanded by the last column and collecting all sub-determinants on the left side of the $\triangle^{i}$ with the action $"\circ"$. And $T_m^{-1}$ is expanded by the first column and all the sub-determinants are on the right side with the action  $"\circ"$.
\end{theorem}

With this representation, the action of $T_m$ on an arbitrary function $q$ is given by the following theorem.
\begin{theorem}
Under the action of $T_m=T_d^{(m)}(q_1^{(m-1)})\circ \dots \circ T_d^{(2)}(q_1^{(1)})\circ T_d^{(1)}(q_1^{(0)})$, the transformed eigenfunctions and the $\tau$-function of the cdKP hierarchy from the arbitrary $L^{(0)}$ are given by
\begin{align}
&q_1^{(m)}=T_m\circ (L^{(0)})^m\cdot q_1^{(0)}= T_m\circ\eta_{m+1}=\frac{W_{m+1}^{\triangle}(\eta_1,\eta_2,\dots,\eta_m,\eta_{m+1})}{W_m^{\triangle}(\eta_1,\eta_2,\dots,\eta_m)},\\
&r_1^{(m)}=\Lambda(q_1^{(m-1)})=\Lambda((T_{m-1}\cdot\eta_m)^{-1})=\frac{\Lambda(W_{m-1}^{\triangle}(\eta_1,\eta_2,\dots,\eta_{m-1}))}{\Lambda(W_{m}^{\triangle}(\eta_1,\eta_2,\dots,\eta_m))},\\
&q_i^{(m)}=T_m\cdot q_i^{(0)}=\frac{ W_{m+1}^{\triangle}(\eta_1,\eta_2,\dots,\eta_m,q_i^{(0)})}{W_m^{\triangle}(\eta_1,\eta_2,\dots,\eta_m)},i=2,\cdots,m,\\
&r_i^{(m)}=(T_m^{-1})^{*}(r_i^{(0)})=(-1)^m\frac{\Lambda( IW_{m+1}^{\triangle}(r_i^{(0)},\eta_1,\eta_2,\dots,\eta_m))}{\Lambda(W_m^{\triangle}(\eta_1,\eta_2,\dots,\eta_m))},i=2,\cdots,m,
\end{align}
and
\begin{equation}
\tau^{(m)}_{\triangle}=W_m^{\triangle}(\eta_1,\eta_2,\dots,\eta_m)\cdot \tau_{\triangle}.
\end{equation} And $\eta_i$ is defined by (\ref{eta}).

\end{theorem}

\section{Conclusions and Discussions}\label{conclusion}

In this paper, we have provided two types of gauge transformation in
Theorem 3.4 and 3.5 of the cdKP hierarchy. The new solutions of
the cdKP hierarchy along single-channel  of difference type $T_d$
are expressed explicitly by the determinants in Theorem 4.2. This result also can be apply to the $k$-constrained dKP hierarchy. These
results gives a powerful tool to construct the explicit solution of
the discrete soliton equations in the cdKP hierarchy. Thus we can
study the particular effects of the discretization of dynamical
variables by comparing specific solutions of the discrete and
continue equations. We shall do it in the future. Moreover, the
presence of the gauge transformation in the cdKP hierarchy shows
again that the discretization used to define the dKP hierarchy is well enough
so that the solvability is inherited by it.

 Different from the generating
functions of the gauge transformation of the dKP hierarchy which can be chosen
freely, the ones of the gauge transformations  of the cdKP hierarchy only
can be selected from the eigenfunctions and the adjoint eigenfunctions
in the Lax operator of the cdKP hierarchy, which  leads to the main
particularity of the gauge transformation of the cdKP hierarchy.

{\bf Acknowledgments} {\noindent \small This work is supported by
the NSF of China under Grant No.10971109 and Science Fund of Ningbo
University (No.xk1062, No.XYL11012), K.C.Wong Magna Fund in
Ningbo University.

The authors thanks the valuable comments of the Referees.
}

%%%%%%%%%%%%%%%%% References  %%%%%%%%%%%%%%%%%%%%%%%%%%%%%%%%%%%%%%%
\newpage{}

%%%%%%%%%%%%%%%%%%%%%%%%%%%%%%%%%%%%%%%%%%%%%%%%%%%%%%%%%

\end{document}